\documentclass[review]{elsarticle}

\usepackage{lineno,hyperref}
\modulolinenumbers[5]

\journal{Journal of \LaTeX\ Templates}
\usepackage[utf8x]{inputenc}
\DeclareMathAlphabet\mbi{OML}{cmm}{b}{it}
\DeclareSymbolFont{boldsymbols}{OMS}{cmsy}{b}{n}
\DeclareSymbolFontAlphabet{\mathbfcal}{boldsymbols}

\usepackage{geometry}
\usepackage{amsthm}
\usepackage{amsmath}
\usepackage{amsbsy}
\usepackage{amsfonts}
\usepackage{dsfont}
\usepackage{mathrsfs}
\everymath{\displaystyle}
\newtheorem{thm}{Theorem}[section]
\newtheorem{corollary}{Corollary}[section]
\theoremstyle{definition}

\theoremstyle{Remarque}

\theoremstyle{plain}

\usepackage{graphicx}
\usepackage{rotating}
\usepackage{marvosym}
\usepackage{textcomp}
\usepackage {setspace}
\usepackage{fancyhdr}
\usepackage{hyphenat}
\usepackage{multirow}
\usepackage{pdfpages}
\usepackage{pdftexcmds}
\usepackage{subfigure}
\usepackage{booktabs,caption,fixltx2e}
\usepackage[flushleft]{threeparttable}
\usepackage{latexsym}
\usepackage{lmodern}
\usepackage{a4wide}
\usepackage{lscape}
\usepackage{rotating}
\usepackage[flushleft]{threeparttable}
\DeclareMathOperator*{\argmax}{arg\,max}
\DeclareMathOperator*{\argmin}{arg\,min}

\biboptions{authoryear}
\usepackage{filecontents}

\begin{document}

\begin{frontmatter}

\title{Profiled deviance for the multivariate linear mixed-effects model fitting}
% Consistent estimates in the multivariate linear mixed-effects model

%\author{Eric Houngla Adjakossa\fnref{myfootnote}}
%\author{Gregory Nuel\fnref{myfootnote}}

\author[mymainaddress,mysecondaryaddress]{Eric Houngla Adjakossa\corref{correspondingauthor}}
\cortext[correspondingauthor]{Corresponding author}
\ead{ericadjakossah@gmail.com}

\author[mysecondaryaddress]{Gregory Nuel}
\ead{Gregory.Nuel@math.cnrs.fr}

\address[mymainaddress]{International Chair in Mathematical Physics and Applications (ICMPA-UNESCO Chair) /University of Abomey-Calavi, 072 B.P. 50 Cotonou, Republic of Benin}
\address[mysecondaryaddress]{ Laboratoire de Probabilit\'es et Mod\`eles Al\'eatoires /Universit\'e Pierre et Marie Curie, Case courrier 188 - 4, Place Jussieu 75252 Paris cedex 05 France}

\begin{abstract}
This paper focuses on the multivariate linear mixed-effects model, including all the correlations between the random effects when the marginal residual terms are assumed uncorrelated and homoscedastic with possibly different standard deviations. The random effects covariance matrix is Cholesky factorized to directly estimate the variance components of these random effects. This strategy enables a consistent estimate of the random effects covariance matrix which, generally, has a poor estimate when it is grossly (or directly) estimated, using the estimating methods such as the EM algorithm. By using simulated data sets, we compare the estimates based on the present method with the EM algorithm-based estimates. We provide an illustration by using the real-life data concerning the study of the child's immune against malaria in Benin (West Africa).
\end{abstract}


\end{frontmatter}

%\linenumbers

\section{Introduction}

Linear mixed-effects model~\citep{hartley1967maximum, laird1982random, verbeke1997linear,hedeker2006longitudinal,fitzmaurice2012applied}  has become a popular tool for analyzing univariate multilevel data which arise in many areas (biology, medicine, economy, etc), due to its flexibility to model the correlation contained in these data, and the availability of reliable and efficient software packages for fitting it~\citep{bates2014lme4, pinheiro2007linear, littell1996random, halekoh2006r}. Univariate multilevel data are referred to as observations (or measurements) of a single variable of interest on several levels (school in a village which, in turn, is in a town), while multivariate multilevel data are characterized by multiple variables of interest measured at multiple levels. Examples include exam or test scores recorded for students across time, and multiple items at a single occasion for students in more than one school. 
Multivariate extension of the (single response variable-based) linear mixed-effects model is, indeed, having increasing popularity as flexible tool for the analysis of multivariate multilevel data~\citep{sammel1999multivariate, schafer2002computational,wang2010ecm, jensen2012implementation}. 

For the linear mixed-effects model, many methods for obtaining the estimates of the fixed and the random effects have been proposed in the literature. These methods include Henderson's mixed model equations~\citep{henderson1950estimation}, approaches proposed by \cite{goldberger1962best} as well as techniques based on two-stage regression, Bayes estimation, etc. For details, see~\cite[][Section 7.4c]{searle1992variance} and~\cite{robinson1991blup}. Concerning the variance parameters estimation in linear mixed-effects model, the discussed methods in the literature include the ANOVA method for balanced data which uses the expected mean squares approach~\citep{searle1995overview, searle1971linear}. For unbalanced data, \cite{rao1971estimation} proposed the minimum norm quadratic estimation (MINQUE) method, where the resulting estimates are translation invariant under unbiased quadratic forms of the observations. \cite{lee1998generalized} gave another method of estimating variance parameters using extended quasi-likelihood, i.e. gamma-log generalized linear models. For more details on these parameters' estimation methods in the linear mixed-effects model, see the paper of~\cite{gumedze2011parameter}.
Beside all the methods cited earlier, come the Maximum Likelihood (ML) and the Restricted Maximum Likelihood (REML) methods. ML and REML methods are the most popular estimation methods in the linear mixed-effects model~\citep{lindstrom1988newton}. The main attraction of these methods is that they can handle a much wider class of variance models than simple variance components~\citep{gumedze2011parameter}.

In the multivariate linear mixed-effects model, ML and REML estimates are frequently approached through iterative schemes such as EM algorithm~\citep{meng1993maximum,dempster1977maximum,an2013latent, schafer2002computational, shah1997random}. This avoid the difficulties related to the direct calculating of the parameters' likelihood, since the random effects are not observed, without ignoring the flexible computationally of these algorithms. Despite the existence of valid theorems which show the asymptotic convergence of the sequences produced by these algorithms toward ML estimates~\citep{dempster1977maximum}, in practice this may not always work exactly as expected.

In this paper, we focus on the multivariate linear mixed-effects model, including all the correlations between the random effects while the marginal residuals are assumed independent homoscedastic with possibly different standard deviation. The class of multivariate mixed-effects models considered here assumes that the random effects and the residuals follow Gaussian distributions, and the dependent variables are continuous. In this model, our approach consists in directly calculating the likelihood of the model's parameters. This likelihood is used to obtain the ML estimates or the REML estimates through the provided REML criterion. This strategy may explain the high quality of the estimates of both fixed effects parameters and random effects' variance parameters as well as residual variance parameters. This approach may be viewed as a generalization of the approach proposed by~\cite{bates2014lme4} under the R software~\citep{r2015} package named lme4.

%%%%%%%%%%%%%%%%%%%%%%%%%%% Multivariate linear mixed-effects model %%%%%%%%%%%%%%%%%%%%%%%%%%%%%%%%%%%%
\section{Multivariate linear mixed-effects model}
For the sake of simplicity we focus on the bivariate case ($d=2$) in most of the paper, but the generalization to higher dimensions ($d>2$) is straightforward. Thus, in dimension 2, the model is the following:

\begin{eqnarray}
y_1=X_1\beta_1+Z_1\gamma_1+\varepsilon_1,\nonumber\\
y_2=X_2\beta_2+Z_2\gamma_2+\varepsilon_2,
\end{eqnarray}
where

\begin{equation}\label{eq:Modelassumptions}
\boldsymbol{\gamma}=
\begin{pmatrix}
\gamma_1\\\gamma_2
\end{pmatrix}
\sim
\mathcal{N}
\left(
\boldsymbol{0}
,
\boldsymbol{\Gamma}=
\begin{pmatrix}
\Gamma_1&\Gamma_{12}\\
\Gamma_{12}^\top&\Gamma_2
\end{pmatrix}
\right)
,
\quad
\boldsymbol{\varepsilon}=
\begin{pmatrix}
\varepsilon_1\\\varepsilon_2
\end{pmatrix}
\sim
\left(
\boldsymbol{0}
,
\begin{pmatrix}
\sigma_1^2\text{I}_{N}&0\\
0&\sigma_2^2\text{I}_{N}
\end{pmatrix}
\right).
\end{equation}

For the sake of simplicity, we write $\boldsymbol{\gamma}\sim\mathcal{N}(\boldsymbol{0},\boldsymbol{\Gamma})$ to mean that $\boldsymbol{\gamma}$ is a realization of a random vector which is $\mathcal{N}(\boldsymbol{0},\boldsymbol{\Gamma})$ distributed.
 For $k\in \{1,2\}$, $\beta_k\text{ and }\gamma_k$ denote respectively the fixed effects and the random effects vector of covariates, while $\varepsilon_k$ is the marginal residual component in the dimension $k$ of the model. $X_k$ is a matrix of covariates and $Z_k$ a covariates-based matrix of design. dim($X_k$)$=N\times p_k$ and dim($Z_k$)$=N\times q_k$, where $N$ is the total number of observations. $p_k$ and $q_k$ are, respectively, the number of fixed effect related covariates and the number of random effect related covariates in the dimension $k$ of the model. $\boldsymbol{y}=(y_1^\top,y_2^\top)^\top$ is the vector of marginal observed response variables of the model. We assume that $\boldsymbol{y}$ is a realization of a random vector $\boldsymbol{\mathcal{Y}}$ and belongs to $\mathbb{R}^{2N}$. The bold symbols represent parameters, or vectors, of multiple dimensions (i.e. $\Gamma_1$ concerns dimension $1$ of the model while $\boldsymbol{\Gamma}$ concerns both dimensions).

$\Gamma_1$ and $\Gamma_2$ are the variance-covariance matrices of $\gamma_1$ and $\gamma_2$, respectively. $\Gamma_1$ and $\Gamma_2$ must be, indeed, positive semidefinite. It is then convenient to express the model in terms of the relative covariance factors, $\Lambda_{\theta_1}$ and $\Lambda_{\theta_2}$, which are $q_1\times q_1$ and $q_2\times q_2$ matrices, respectively. $\Lambda_{\theta_1}$ is a block diagonal matrix. Each element in the diagonal of $\Lambda_{\theta_1}$ is a lower triangular matrix whose nonzero entries are the components of the vector $\theta_1$. That is, $\theta_1$ generates the symmetric $q_1\times q_1$ variance-covariance matrix $\Gamma_1$, according to 

\begin{equation}\label{eq:Gamma1Decompose}
\Gamma_1=\sigma_1^2\Lambda_{\theta_1}\Lambda_{\theta_1}^\top.
\end{equation}
Same as $\theta_2$ which generates $\Gamma_2$ according to 

\begin{equation}\label{eq:Gamma2Decompose}
\Gamma_2=\sigma_2^2\Lambda_{\theta_2}\Lambda_{\theta_2}^\top.
\end{equation}
In Equations~\ref{eq:Gamma1Decompose} and~\ref{eq:Gamma2Decompose}, $\sigma_1^2$ and $\sigma_2^2$ are the same marginal residual variances used in the model expression (see Equation~\ref{eq:Modelassumptions}). Using the variance-component parameters, $\theta_1$ and $\theta_2$, the marginal random effects, $\gamma_1$ and $\gamma_2$, are expressed as

\begin{equation}
\gamma_1=\Lambda_{\theta_1}u_1,
\quad
\gamma_2=\Lambda_{\theta_2}u_2,
\end{equation}
such that

\begin{equation}\label{eq:u}
\boldsymbol{u}=
\begin{pmatrix}
u_1\\u_2
\end{pmatrix}
\sim
\mathcal{N}
\left(
\boldsymbol{0}
,
\Sigma_{\boldsymbol{u}}
\right)
,
\quad
\text{ with }
\quad
\Sigma_{\boldsymbol{u}}=
\begin{pmatrix}
\sigma_1^2\text{I}_{q_1}&\sigma_1\sigma_2\boldsymbol{\rho}\\
\sigma_1\sigma_2\boldsymbol{\rho}^\top&\sigma_2^2\text{I}_{q_2}
\end{pmatrix}.
\end{equation}
In Equation~\ref{eq:u}, $\boldsymbol{\rho}$ is a block diagonal matrix and $\boldsymbol{u}$ is a realization of a random vector $\boldsymbol{\mathcal{U}}$. The diagonal elements of $\boldsymbol{\rho}$, say $\rho$, are matrices which contain the correlations between $\gamma_1$ and $\gamma_2$. For example, if $\gamma_1=(\gamma_1^\text{I},\gamma_1^\text{S})^\top$ and $\gamma_2=(\gamma_2^\text{I},\gamma_2^\text{S})^\top$, with $\text{I}=\text{Intercept}$ and $\text{S}=\text{Slope}$, 

\begin{equation}
\rho=
\begin{pmatrix}
\text{corr}(\gamma_1^\text{I},\gamma_2^\text{I})&\text{corr}(\gamma_1^\text{I},\gamma_2^\text{S})\\
\text{corr}(\gamma_1^\text{S},\gamma_2^\text{I})&\text{corr}(\gamma_1^\text{S},\gamma_2^\text{S})
\end{pmatrix}
\quad
\text{ and }
\quad
\boldsymbol{\rho}=
\text{diag}(\rho,\dots,\rho).
\end{equation}
The bivariate linear mixed-effects model is then re-expressed as:

\begin{eqnarray}\label{eq:model_rewriten1}
y_1=X_1\beta_1+Z_1\Lambda_{\theta_1}u_1+\varepsilon_1,\nonumber\\
y_2=X_2\beta_2+Z_2\Lambda_{\theta_2}u_2+\varepsilon_2,
\end{eqnarray}
with

\begin{equation}\label{eq:model_rewriten2}
\boldsymbol{u}=
\begin{pmatrix}
u_1\\u_2
\end{pmatrix}
\sim
\mathcal{N}
\left(
\boldsymbol{0}
,
\Sigma_{\boldsymbol{u}}
\right)
,
\quad
\boldsymbol{\varepsilon}=
\begin{pmatrix}
\varepsilon_1\\\varepsilon_2
\end{pmatrix}
\sim
\left(
\boldsymbol{0}
,
\begin{pmatrix}
\sigma_1^2\text{I}_{N}&0\\
0&\sigma_2^2\text{I}_{N}
\end{pmatrix}
\right).
\end{equation}
Then the parameters which will be estimated are $\beta_1, \beta_2, \sigma_1^2, \sigma_2^2, \theta_1$, $\theta_2$ and $\rho$.
%%%%%%%%%%%%%%%%%%%%%%%%%%% End Multivariate linear mixed-effects model %%%%%%%%%%%%%%%%%%%%%%%%%%%%%%%%%%

\section{Parameters' estimates}
In this Section, we first provide the likelihood of the model's parameters and then give the REML criterion which will be optimized for the obtaining of the parameters' REML estimates.
%%%%%%%%%%%%%%%%%%%%%%%%%%% ML criterion %%%%%%%%%%%%%%%%%%%%%%%%%%%%%%%%%%%%%%%%%%%%%%%%
\subsection{ML criterion}
The ML criterion is the log-likelihood of the model's parameters which is displayed through the following theorem

\begin{thm}\label{theorem1}
Suppose that $\boldsymbol{y}=(y_1^\top,y_2^\top)^\top$ satisfies the bivariate linear mixed-effects model expressed by Equations~(\ref{eq:model_rewriten1} and \ref{eq:model_rewriten2}), where $\beta_1, \beta_2, \sigma_1^2, \sigma_2^2, \theta_1$, $\theta_2$, $\rho$ are the parameters which need to be estimated, and $\boldsymbol{\beta}=(\beta_1^\top,\beta_2^\top)^\top$, $\boldsymbol{\sigma}=(\sigma_1^2,\sigma_2^2)^\top$, $\boldsymbol{\theta}=(\theta_1^\top,\theta_2^\top)^\top$. Denoting by $Y_{\boldsymbol{\sigma}}=\left(\sqrt{\sigma_2^2}y_1^\top,\sqrt{\sigma_1^2}y_2^\top\right)^\top$, $X_{\boldsymbol{\sigma}}=\begin{pmatrix}\sqrt{\sigma_2^2}X_1&\boldsymbol{0}\\\boldsymbol{0}&\sqrt{\sigma_1^2}X_2\end{pmatrix}$, $Z_{\boldsymbol{\sigma}\boldsymbol{\theta}}=
\begin{pmatrix}\sqrt{\sigma_2^2}Z_1\Lambda_{\theta_1}&\boldsymbol{0}\\\boldsymbol{0}&\sqrt{\sigma_1^2}Z_2\Lambda_{\theta_2}\end{pmatrix}$, and $\mu_{\boldsymbol{\mathcal{U}}|\boldsymbol{\mathcal{Y}}=\boldsymbol{y}}$  the conditional mean of $\boldsymbol{\mathcal{U}}$ given that $\boldsymbol{\mathcal{Y}}=\boldsymbol{y}$, the log-likelihood of $\boldsymbol{\beta}$, $\boldsymbol{\sigma}$, $\boldsymbol{\theta}$ and $\rho$ given $\boldsymbol{y}$ is expressed as

\begin{eqnarray}
\ell(\boldsymbol{\beta},\boldsymbol{\theta},\rho,\boldsymbol{\sigma}|\boldsymbol{y})
&=&-\frac{r(\widehat{\boldsymbol{\beta}}_{\boldsymbol{\theta},\rho,\boldsymbol{\sigma}},\mu_{\boldsymbol{\mathcal{U}}|\boldsymbol{\mathcal{Y}}=\boldsymbol{y}})+\left\|R_X(\boldsymbol{\beta}-\widehat{\boldsymbol{\beta}}_{\boldsymbol{\theta},\rho,\boldsymbol{\sigma}})\right\|^2}{2\sigma_1^2\sigma_2^2}-\frac{N-q}{2}\log(\sigma_1^2\sigma_2^2)\nonumber\\
&&-\frac{1}{2}\log(|\Sigma_{\boldsymbol{u}}|)-\frac{1}{2}\log(|L_{\boldsymbol{\theta},\rho,\boldsymbol{\sigma}}|^2),
\end{eqnarray}
where $q=q_1+q_2$, $\widehat{\boldsymbol{\beta}}_{\boldsymbol{\theta},\rho,\boldsymbol{\sigma}}$ and $\mu_{\boldsymbol{\mathcal{U}}|\boldsymbol{\mathcal{Y}}=\boldsymbol{y}}$ satisfy

\begin{equation}
\begin{pmatrix}
X_{\boldsymbol{\sigma}}^\top X_{\boldsymbol{\sigma}}&X_{\boldsymbol{\sigma}}^\top Z_{\boldsymbol{\sigma}\boldsymbol{\theta}}\\
Z_{\boldsymbol{\sigma}\boldsymbol{\theta}}^\top X_\sigma&Z_{\boldsymbol{\sigma}\boldsymbol{\theta}}^\top Z_{\boldsymbol{\sigma}\boldsymbol{\theta}}+\sqrt{\sigma_1^2\sigma_2^2}\Sigma_{\boldsymbol{u}}^{-1}
\end{pmatrix}
\begin{pmatrix}\widehat{\boldsymbol{\beta}}_{\boldsymbol{\theta},\rho,\boldsymbol{\sigma}}\\ \mu_{\mathcal{U}|\mathcal{Y}=\boldsymbol{y}}\end{pmatrix}
=
\begin{pmatrix}
X_{\boldsymbol{\sigma}}^\top\\Z_{\boldsymbol{\sigma}\boldsymbol{\theta}}^\top
\end{pmatrix}
Y_{\boldsymbol{\sigma}},
\end{equation}

\begin{equation}
r(\widehat{\boldsymbol{\beta}}_{\boldsymbol{\theta},\rho,\boldsymbol{\sigma}},\mu_{\boldsymbol{\mathcal{U}}|\boldsymbol{\mathcal{Y}}=\boldsymbol{y}})=\left\|Y_{\boldsymbol{\sigma}}-X_{\boldsymbol{\sigma}}\widehat{\boldsymbol{\beta}}_{\boldsymbol{\theta},\rho,\boldsymbol{\sigma}}-Z_{\boldsymbol{\sigma}\boldsymbol{\theta}}\mu_{\boldsymbol{\mathcal{U}}|\boldsymbol{\mathcal{Y}}=\boldsymbol{y}}\right\|^2 +\sigma_1^2\sigma_2^2\mu_{\boldsymbol{\mathcal{U}}|\boldsymbol{\mathcal{Y}}=\boldsymbol{y}}^\top\Sigma_{\boldsymbol{u}}^{-1}\mu_{\boldsymbol{\mathcal{U}}|\boldsymbol{\mathcal{Y}}=\boldsymbol{y}},
\end{equation}
$L_{\boldsymbol{\theta},\rho,\boldsymbol{\sigma}}$ satisfies

\begin{equation}
L_{\boldsymbol{\theta},\rho,\boldsymbol{\sigma}}L_{\boldsymbol{\theta},\rho,\boldsymbol{\sigma}}^\top=
Z_{\boldsymbol{\sigma}\boldsymbol{\theta}}^\top Z_{\boldsymbol{\sigma}\boldsymbol{\theta}}+\sqrt{\sigma_1^2\sigma_2^2}\Sigma_{\boldsymbol{u}}^{-1},
\end{equation}
and $R_X$ satisfies

\begin{equation}
\begin{pmatrix}
X_{\boldsymbol{\sigma}}^\top X_{\boldsymbol{\sigma}}&X_{\boldsymbol{\sigma}}^\top Z_{\boldsymbol{\sigma}\boldsymbol{\theta}}\\
Z_{\boldsymbol{\sigma}\boldsymbol{\theta}}^\top X_\sigma&L_{\boldsymbol{\theta},\rho,\boldsymbol{\sigma}}L_{\boldsymbol{\theta},\rho,\boldsymbol{\sigma}}^\top
\end{pmatrix}
=
\begin{pmatrix}
R_X&\boldsymbol{0}\\R_{ZX}&L_{\boldsymbol{\theta},\rho,\boldsymbol{\sigma}}^\top
\end{pmatrix}^\top
\begin{pmatrix}
R_X&\boldsymbol{0}\\R_{ZX}&L_{\boldsymbol{\theta},\rho,\boldsymbol{\sigma}}^\top
\end{pmatrix}.
\end{equation}

\end{thm}

\begin{proof}
Denoting by $f_{\mathcal{X}}(.)$ the density function of any random vector $\mathcal{X}$, 

\begin{equation}
f_{\boldsymbol{\mathcal{Y}}}(\boldsymbol{y})=\int_{\mathbb{R}^{q_1+q_2}}f_{\boldsymbol{\mathcal{Y}},\boldsymbol{\mathcal{U}}}(\boldsymbol{y},\boldsymbol{u})d\boldsymbol{u},
\end{equation}
where

\begin{eqnarray}
f_{\boldsymbol{\mathcal{Y}},\boldsymbol{\mathcal{U}}}(\boldsymbol{y},\boldsymbol{u})&=&f_{\boldsymbol{\mathcal{Y}}|\boldsymbol{\mathcal{U}}}(\boldsymbol{y}|\boldsymbol{u})f_{\boldsymbol{\mathcal{U}}}(\boldsymbol{u})=f_{\mathcal{Y}_1|\mathcal{U}_1}(y_1|u_1)f_{\mathcal{Y}_2|\mathcal{U}_2}(y_2|u_2)f_{\boldsymbol{\mathcal{U}}}(\boldsymbol{u})\nonumber\\
&=&(2\pi\sigma_1^2)^{-\frac{N}{2}}(2\pi\sigma_2^2)^{-\frac{N}{2}}(2\pi)^{-\frac{q_1+q_2}{2}}|\Sigma_{\boldsymbol{u}}|^{-\frac{1}{2}}\exp\left(-\frac{\|y_1-X_1\beta_1-Z_1\Lambda_{\theta_1}u_1\|^2}{2\sigma_1^2}\right.\nonumber\\
&&\left.-\frac{\|y_2-X_2\beta_2-Z_2\Lambda_{\theta_2}u_2\|^2}{2\sigma_2^2}-\frac{1}{2}\boldsymbol{u}^\top\Sigma_{\boldsymbol{u}}^{-1}\boldsymbol{u}\right).
\end{eqnarray}
Let us denote by $\widetilde{\Sigma}$ the matrix such that 

\begin{equation}
\Sigma_{\boldsymbol{u}}^{-1}=\widetilde{\Sigma}^\top\widetilde{\Sigma}.
\end{equation}
It then comes that $\boldsymbol{u}^\top\Sigma_{\boldsymbol{u}}^{-1}\boldsymbol{u}=\|\widetilde{\Sigma}\boldsymbol{u}\|^2$ and

\begin{eqnarray}
&&\frac{\|y_1-X_1\beta_1-Z_1\Lambda_{\theta_1}u_1\|^2}{\sigma_1^2}+\frac{\|y_2-X_2\beta_2-Z_2\Lambda_{\theta_2}u_2\|^2}{\sigma_2^2}+\boldsymbol{u}^\top\Sigma_{\boldsymbol{u}}^{-1}\boldsymbol{u}\nonumber\\
&=&\frac{\|\sqrt{\sigma_2^2}(y_1-X_1\beta_1-Z_1\Lambda_{\theta_1}u_1)\|^2+\|\sqrt{\sigma_1^2}(y_2-X_2\beta_2-Z_2\Lambda_{\theta_2}u_2)\|^2+\|\sqrt{\sigma_1^2\sigma_2^2}\widetilde{\Sigma}\boldsymbol{u}\|^2}{\sigma_1^2\sigma_2^2}\nonumber\\
&=&\left\|\begin{pmatrix}\sqrt{\sigma_2^2}y_1\\\sqrt{\sigma_1^2}y_2\\\boldsymbol{0}_{q_1+q_2}\end{pmatrix}-\begin{pmatrix}\sqrt{\sigma_2^2}X_1&\boldsymbol{0}_{Np_2}&\sqrt{\sigma_2^2}Z_1\Lambda_{\theta_1}&\boldsymbol{0}_{Nq_2}\\\boldsymbol{0}_{Np_1}&\sqrt{\sigma_1^2}X_2&\boldsymbol{0}_{Nq_1}&\sqrt{\sigma_1^2}Z_2\Lambda_{\theta_2}\\&\boldsymbol{0}_{q_1+q_2,p_1+p_2}&&\sqrt{\sigma_1^2\sigma_2^2}\tilde{\Sigma}\end{pmatrix}\begin{pmatrix}\boldsymbol{\beta}\\ \boldsymbol{u}\end{pmatrix} \right\|^2\\
&=&\left\| Y_\Lambda-Z_{X\Lambda}\begin{pmatrix}\boldsymbol{\beta}\\ \boldsymbol{u}\end{pmatrix}\right\|^2\\
&=&g(\boldsymbol{\beta},\boldsymbol{u},\boldsymbol{\theta},\rho,\boldsymbol{\sigma}).
\end{eqnarray}

\begin{eqnarray}
\begin{pmatrix}\widehat{\boldsymbol{\beta}}_{\boldsymbol{\theta},\rho,\boldsymbol{\sigma}}\\ \mu_{\boldsymbol{\mathcal{U}}|\boldsymbol{\mathcal{Y}}=\boldsymbol{y}}\end{pmatrix}=\argmin_{\boldsymbol{u},\boldsymbol{\beta}} g(\boldsymbol{\beta},\boldsymbol{u},\boldsymbol{\theta},\rho,\boldsymbol{\sigma})&\iff& Z_{X\Lambda}^\top Z_{X\Lambda}\begin{pmatrix}\widehat{\boldsymbol{\beta}}_{\boldsymbol{\theta},\rho,\boldsymbol{\sigma}}\\ \mu_{\boldsymbol{\mathcal{U}}|\boldsymbol{\mathcal{Y}}=\boldsymbol{y}}\end{pmatrix}=Z_{X\Lambda}^\top Y_\Lambda \text{ (normal eq.)}\label{normal_equation},\nonumber\\
\end{eqnarray}
with

\begin{equation}
Z_{X\Lambda}^\top Z_{X\Lambda}=
\begin{pmatrix}
X_{\boldsymbol{\sigma}}^\top X_{\boldsymbol{\sigma}}&X_{\boldsymbol{\sigma}}^\top Z_{\boldsymbol{\sigma}\boldsymbol{\theta}}\\
Z_{\boldsymbol{\sigma}\boldsymbol{\theta}}^\top X_\sigma&Z_{\boldsymbol{\sigma}\boldsymbol{\theta}}^\top Z_{\boldsymbol{\sigma}\boldsymbol{\theta}}+\sqrt{\sigma_1^2\sigma_2^2}\Sigma_{\boldsymbol{u}}^{-1}
\end{pmatrix}
\quad
\text{ and }
\quad
Z_{X\Lambda}^\top Y_\Lambda=
\begin{pmatrix}
X_{\boldsymbol{\sigma}}^\top\\Z_{\boldsymbol{\sigma}}^\top
\end{pmatrix}
Y_{\boldsymbol{\sigma}}.
\end{equation}
By setting $p=p_1+p_2$, $dim(Z_{X\Lambda})=(2N+q)\times(p+q)$ and $S=Im(Z_{X\Lambda})$ is a subspace of $\mathbb{R}^{2N+q}$. $Y_\Lambda\in\mathbb{R}^{2N+q}$ and $Z_{X\Lambda}\begin{pmatrix}\widehat{\boldsymbol{\beta}}_{\boldsymbol{\theta},\rho,\boldsymbol{\sigma}}\\ \mu_{\boldsymbol{\mathcal{U}}|\boldsymbol{\mathcal{Y}}=\boldsymbol{y}}\end{pmatrix}$ is the orthogonal projection of $Y_\Lambda$ on $S$. Then, 

\begin{equation}
Z_{X\Lambda}u\perp\left[Y_\Lambda-Z_{X\Lambda}\begin{pmatrix}\widehat{\boldsymbol{\beta}}_{\boldsymbol{\theta},\rho,\boldsymbol{\sigma}}\\ \mu_{\boldsymbol{\mathcal{U}}|\boldsymbol{\mathcal{Y}}=\boldsymbol{y}}\end{pmatrix}\right], \forall u\in \mathbb{R}^{p+q}.
\end{equation}
And $g(\boldsymbol{\beta},\boldsymbol{u},\boldsymbol{\theta},\rho,\boldsymbol{\sigma})$ can then be rewritten as:

\begin{eqnarray}
g(\boldsymbol{\beta},\boldsymbol{u},\boldsymbol{\theta},\rho,\boldsymbol{\sigma})&=&\left\| Y_\Lambda-Z_{X\Lambda}\begin{pmatrix}\boldsymbol{\beta}\\\boldsymbol{u}\end{pmatrix}+Z_{X\Lambda}\begin{pmatrix}\widehat{\boldsymbol{\beta}}_{\boldsymbol{\theta},\rho,\boldsymbol{\sigma}}\\ \mu_{\boldsymbol{\mathcal{U}}|\boldsymbol{\mathcal{Y}}=\boldsymbol{y}}\end{pmatrix}-Z_{X\Lambda}\begin{pmatrix}\widehat{\boldsymbol{\beta}}_{\boldsymbol{\theta},\rho,\boldsymbol{\sigma}}\\ \mu_{\boldsymbol{\mathcal{U}}|\boldsymbol{\mathcal{Y}}=\boldsymbol{y}}\end{pmatrix} \right\|^2\\
&=&\left\|Y_\Lambda-Z_{X\Lambda}\begin{pmatrix}\widehat{\boldsymbol{\beta}}_{\boldsymbol{\theta},\rho,\boldsymbol{\sigma}}\\ \mu_{\boldsymbol{\mathcal{U}}|\boldsymbol{\mathcal{Y}}=\boldsymbol{y}}\end{pmatrix}  \right\|^2+\left\|Z_{X\Lambda} \begin{pmatrix}\boldsymbol{\beta}-\widehat{\boldsymbol{\beta}}_{\boldsymbol{\theta},\rho,\boldsymbol{\sigma}}\\\boldsymbol{u}- \mu_{\boldsymbol{\mathcal{U}}|\boldsymbol{\mathcal{Y}}=\boldsymbol{y}}\end{pmatrix} \right\|^2\\
&=&\left\|Y_\Lambda-Z_{X\Lambda}\begin{pmatrix}\widehat{\boldsymbol{\beta}}_{\boldsymbol{\theta},\rho,\boldsymbol{\sigma}}\\ \mu_{\boldsymbol{\mathcal{U}}|\boldsymbol{\mathcal{Y}}=\boldsymbol{y}}\end{pmatrix}  \right\|^2+\begin{pmatrix}\boldsymbol{\beta}-\widehat{\boldsymbol{\beta}}_{\boldsymbol{\theta},\rho,\boldsymbol{\sigma}}\\\boldsymbol{u}- \mu_{\boldsymbol{\mathcal{U}}|\boldsymbol{\mathcal{Y}}=\boldsymbol{y}}\end{pmatrix}^\top Z_{X\Lambda}^\top Z_{X\Lambda} \begin{pmatrix}\boldsymbol{\beta}-\widehat{\boldsymbol{\beta}}_{\boldsymbol{\theta},\rho,\boldsymbol{\sigma}}\\\boldsymbol{u}- \mu_{\boldsymbol{\mathcal{U}}|\boldsymbol{\mathcal{Y}}=\boldsymbol{y}}\end{pmatrix}.\nonumber\\
\end{eqnarray}
$Z_{X\Lambda}^\top Z_{X\Lambda}$ can be Cholesky decomposed as

\begin{equation}
Z_{X\Lambda}^\top Z_{X\Lambda}=
\begin{pmatrix}
R_X&\boldsymbol{0}\\R_{ZX}&L_{\boldsymbol{\theta},\rho,\boldsymbol{\sigma}}^\top
\end{pmatrix}^\top
\begin{pmatrix}
R_X&\boldsymbol{0}\\R_{ZX}&L_{\boldsymbol{\theta},\rho,\boldsymbol{\sigma}}^\top
\end{pmatrix},
\end{equation}
where 

\begin{equation}
L_{\boldsymbol{\theta},\rho,\boldsymbol{\sigma}}L_{\boldsymbol{\theta},\rho,\boldsymbol{\sigma}}^\top=
Z_{\boldsymbol{\sigma}\boldsymbol{\theta}}^\top Z_{\boldsymbol{\sigma}\boldsymbol{\theta}}+\sqrt{\sigma_1^2\sigma_2^2}\Sigma_{\boldsymbol{u}}^{-1}.
\end{equation}
Thereafter,

\begin{eqnarray}
g(\boldsymbol{\beta},\boldsymbol{u},\boldsymbol{\theta},\rho,\boldsymbol{\sigma})&=&\left\|Y_{\boldsymbol{\sigma}}-X_{\boldsymbol{\sigma}}\widehat{\boldsymbol{\beta}}_{\boldsymbol{\theta},\rho,\boldsymbol{\sigma}}-Z_{\boldsymbol{\sigma}\boldsymbol{\theta}}\mu_{\boldsymbol{\mathcal{U}}|\boldsymbol{\mathcal{Y}}=\boldsymbol{y}}\right\|^2 +\sigma_1^2\sigma_2^2\mu_{\boldsymbol{\mathcal{U}}|\boldsymbol{\mathcal{Y}}=\boldsymbol{y}}^\top\Sigma_{\boldsymbol{u}}^{-1}\mu_{\boldsymbol{\mathcal{U}}|\boldsymbol{\mathcal{Y}}=\boldsymbol{y}}\nonumber\\
&&+ \left\|R_X(\boldsymbol{\beta}-\widehat{\boldsymbol{\beta}}_{\boldsymbol{\theta},\rho,\boldsymbol{\sigma}}) \right\|^2+ \left\|R_{ZX}(\boldsymbol{\beta}-\widehat{\boldsymbol{\beta}}_{\boldsymbol{\theta},\rho,\boldsymbol{\sigma}})+L_{\boldsymbol{\theta},\rho,\boldsymbol{\sigma}}^\top(\boldsymbol{u}- \mu_{\boldsymbol{\mathcal{U}}|\boldsymbol{\mathcal{Y}}=\boldsymbol{y}}) \right\|^2.\nonumber\\
\end{eqnarray}
By setting 
\begin{equation}
r(\widehat{\boldsymbol{\beta}}_{\boldsymbol{\theta},\rho,\boldsymbol{\sigma}},\mu_{\boldsymbol{\mathcal{U}}|\boldsymbol{\mathcal{Y}}=\boldsymbol{y}})=\left\|Y_{\boldsymbol{\sigma}}-X_{\boldsymbol{\sigma}}\widehat{\boldsymbol{\beta}}_{\boldsymbol{\theta},\rho,\boldsymbol{\sigma}}-Z_{\boldsymbol{\sigma}\boldsymbol{\theta}}\mu_{\boldsymbol{\mathcal{U}}|\boldsymbol{\mathcal{Y}}=\boldsymbol{y}}\right\|^2 +\sigma_1^2\sigma_2^2\mu_{\boldsymbol{\mathcal{U}}|\boldsymbol{\mathcal{Y}}=\boldsymbol{y}}^\top\Sigma_{\boldsymbol{u}}^{-1}\mu_{\boldsymbol{\mathcal{U}}|\boldsymbol{\mathcal{Y}}=\boldsymbol{y}},
\end{equation}
and returning to the calculation of $f_{\boldsymbol{\mathcal{Y}}}(\boldsymbol{y})$, it comes

\begin{eqnarray}
f_{\boldsymbol{\mathcal{Y}}}(\boldsymbol{y})&=&\frac{\int\exp\left[-\frac{r(\widehat{\boldsymbol{\beta}}_{\boldsymbol{\theta},\rho,\boldsymbol{\sigma}},\mu_{\boldsymbol{\mathcal{U}}|\boldsymbol{\mathcal{Y}}=\boldsymbol{y}})+\left\|R_X(\boldsymbol{\beta}-\widehat{\boldsymbol{\beta}}_{\boldsymbol{\theta},\rho,\boldsymbol{\sigma}})\right\|^2 + \left\|R_{ZX}(\boldsymbol{\beta}-\widehat{\boldsymbol{\beta}}_{\boldsymbol{\theta},\rho,\boldsymbol{\sigma}})+L_{\boldsymbol{\theta},\rho,\boldsymbol{\sigma}}^\top(\boldsymbol{u}- \mu_{\boldsymbol{\mathcal{U}}|\boldsymbol{\mathcal{Y}}=\boldsymbol{y}})\right\|^2}{2\sigma_1^2\sigma_2^2}\right]d\boldsymbol{u}}{(2\pi\sigma_1^2)^{N/2}(2\pi\sigma_2^2)^{N/2}(2\pi)^{q/2}|\Sigma_{\boldsymbol{u}}|^{1/2}}\nonumber\\
&=&\frac{\exp\left[-\frac{r(\widehat{\boldsymbol{\beta}}_{\boldsymbol{\theta},\rho,\boldsymbol{\sigma}},\mu_{\boldsymbol{\mathcal{U}}|\boldsymbol{\mathcal{Y}}=\boldsymbol{y}})+\left\|R_X(\boldsymbol{\beta}-\widehat{\boldsymbol{\beta}}_{\boldsymbol{\theta},\rho,\boldsymbol{\sigma}})\right\|^2}{2\sigma_1^2\sigma_2^2}\right]}{(2\pi\sigma_1^2)^{N/2}(2\pi\sigma_2^2)^{N/2}(2\pi)^{q/2}|\Sigma_{\boldsymbol{u}}|^{1/2}}\times\nonumber\\
&&\int\exp\left[-\frac{\left\|R_{ZX}(\boldsymbol{\beta}-\widehat{\boldsymbol{\beta}}_{\boldsymbol{\theta},\rho,\boldsymbol{\sigma}})+L_{\boldsymbol{\theta},\rho,\boldsymbol{\sigma}}^\top(\boldsymbol{u}- \mu_{\boldsymbol{\mathcal{U}}|\boldsymbol{\mathcal{Y}}=\boldsymbol{y}})\right\|^2}{2\sigma_1^2\sigma_2^2}\right]d\boldsymbol{u}.\nonumber\\
\end{eqnarray}
By setting $\boldsymbol{v}=R_{ZX}(\boldsymbol{\beta}-\widehat{\boldsymbol{\beta}}_{\boldsymbol{\theta},\rho,\boldsymbol{\sigma}})+L_{\boldsymbol{\theta},\rho,\boldsymbol{\sigma}}^\top(\boldsymbol{u}- \mu_{\boldsymbol{\mathcal{U}}|\boldsymbol{\mathcal{Y}}=\boldsymbol{y}})$, $d\boldsymbol{u}=\frac{1}{|L_{\boldsymbol{\theta},\rho,\boldsymbol{\sigma}}|}d\boldsymbol{v}$ and

\begin{eqnarray}
f_{\boldsymbol{\mathcal{Y}}}(\boldsymbol{y})&=&\frac{\exp\left[-\frac{r(\widehat{\boldsymbol{\beta}}_{\boldsymbol{\theta},\rho,\boldsymbol{\sigma}},\mu_{\boldsymbol{\mathcal{U}}|\boldsymbol{\mathcal{Y}}=\boldsymbol{y}})+\left\|R_X(\boldsymbol{\beta}-\widehat{\boldsymbol{\beta}}_{\boldsymbol{\theta},\rho,\boldsymbol{\sigma}})\right\|^2}{2\sigma_1^2\sigma_2^2}\right](\sigma_1^2\sigma_2^2)^{\frac{q}{2}}}{(2\pi\sigma_1^2)^{N/2}(2\pi\sigma_2^2)^{N/2}|\Sigma_{\boldsymbol{u}}|^{1/2}|L_{\boldsymbol{\theta},\rho,\boldsymbol{\sigma}}|}\int\frac{1}{(2\pi\sigma_1^2\sigma_2^2)^\frac{q}{2}}\exp\left[-\frac{\left\|\boldsymbol{v}\right\|^2}{2\sigma_1^2\sigma_2^2}\right]d\boldsymbol{v}\nonumber\\
&=&\frac{\exp\left[-\frac{r(\widehat{\boldsymbol{\beta}}_{\boldsymbol{\theta},\rho,\boldsymbol{\sigma}},\mu_{\boldsymbol{\mathcal{U}}|\boldsymbol{\mathcal{Y}}=\boldsymbol{y}})+\left\|R_X(\boldsymbol{\beta}-\widehat{\boldsymbol{\beta}}_{\boldsymbol{\theta},\rho,\boldsymbol{\sigma}})\right\|^2}{2\sigma_1^2\sigma_2^2}\right](\sigma_1^2\sigma_2^2)^{\frac{q}{2}}}{(2\pi\sigma_1^2)^{N/2}(2\pi\sigma_2^2)^{N/2}|\Sigma_{\boldsymbol{u}}|^{1/2}|L_{\boldsymbol{\theta},\rho,\boldsymbol{\sigma}}|}.
\end{eqnarray}
The log-likelihood to be maximized can therefore be expressed as,

\begin{eqnarray}
\ell(\boldsymbol{\beta},\boldsymbol{\theta},\rho,\boldsymbol{\sigma}|\boldsymbol{y})
&=&-\frac{r(\widehat{\boldsymbol{\beta}}_{\boldsymbol{\theta},\rho,\boldsymbol{\sigma}},\mu_{\boldsymbol{\mathcal{U}}|\boldsymbol{\mathcal{Y}}=\boldsymbol{y}})+\left\|R_X(\boldsymbol{\beta}-\widehat{\boldsymbol{\beta}}_{\boldsymbol{\theta},\rho,\boldsymbol{\sigma}})\right\|^2}{2\sigma_1^2\sigma_2^2}-\frac{N-q}{2}\log(\sigma_1^2\sigma_2^2)\nonumber\\
&&-\frac{1}{2}\log(|\Sigma_{\boldsymbol{u}}|)-\frac{1}{2}\log(|L_{\boldsymbol{\theta},\rho,\boldsymbol{\sigma}}|^2).
\end{eqnarray}
\end{proof}
By profiling out $\boldsymbol{\beta}$, the partially profiled log-likelihood is

\begin{eqnarray}
\tilde{\ell}(\boldsymbol{\theta},\rho,\boldsymbol{\sigma}|\boldsymbol{y})&=&-\frac{r(\widehat{\boldsymbol{\beta}}_{\boldsymbol{\theta},\rho,\boldsymbol{\sigma}},\mu_{\boldsymbol{\mathcal{U}}|\boldsymbol{\mathcal{Y}}=\boldsymbol{y}})}{2\sigma_1^2\sigma_2^2}-\frac{N-q}{2}\log(\sigma_1^2\sigma_2^2)\nonumber\\
&&-\frac{1}{2}\log(|\Sigma_{\boldsymbol{u}}|)-\frac{1}{2}\log(|L_{\boldsymbol{\theta},\rho,\boldsymbol{\sigma}}|^2),
\end{eqnarray}
replacing $\widehat{\boldsymbol{\beta}}_{\boldsymbol{\theta},\rho,\boldsymbol{\sigma}}$ by $\boldsymbol{\beta}$. Then, the partially profiled deviance comes

\begin{eqnarray}\label{eq:devianceML}
-2\tilde{\ell}(\boldsymbol{\theta},\rho,\boldsymbol{\sigma}|\boldsymbol{y})&=&\frac{r(\widehat{\boldsymbol{\beta}}_{\boldsymbol{\theta},\rho,\boldsymbol{\sigma}},\mu_{\boldsymbol{\mathcal{U}}|\boldsymbol{\mathcal{Y}}=\boldsymbol{y}})}{\sigma_1^2\sigma_2^2}+(N-q)\log(\sigma_1^2\sigma_2^2)\nonumber\\
&&+\log(|\Sigma_{\boldsymbol{u}}|)+\log(|L_{\boldsymbol{\theta},\rho,\boldsymbol{\sigma}}|^2).
\end{eqnarray}
This deviance is finally the criterion which will be minimized to obtaining the ML estimates of the parameters.

\begin{corollary}
Suppose that $\boldsymbol{y}=(y_1^\top,y_2^\top)^\top$ satisfies the bivariate linear mixed-effects model expressed by Equations~(\ref{eq:model_rewriten1} and \ref{eq:model_rewriten2}). Taking into account the notations in the Theorem~\ref{theorem1}, the ML estimators $\widehat{\boldsymbol{\beta}}$, $\widehat{\boldsymbol{\sigma}}$, $\widehat{\boldsymbol{\theta}}$, $\widehat{\rho}$ of $\boldsymbol{\beta}$, $\boldsymbol{\sigma}$, $\boldsymbol{\theta}$ and $\rho$ satisfy

\begin{equation}
\left(\widehat{\boldsymbol{\theta}},\widehat{\rho},\widehat{\boldsymbol{\sigma}}\right)=
\argmax_{\boldsymbol{\theta},\rho,\boldsymbol{\sigma}} \tilde{\ell}(\boldsymbol{\theta},\rho,\boldsymbol{\sigma}|\boldsymbol{y})
\quad
\text{ and }
\quad
\widehat{\boldsymbol{\beta}}=
\widehat{\boldsymbol{\beta}}_{\widehat{\boldsymbol{\theta}},\widehat{\rho},\widehat{\boldsymbol{\sigma}}}.
\end{equation}

\end{corollary}
%%%%%%%%%%%%%%%%%%%%%%%%%%% End ML criterion %%%%%%%%%%%%%%%%%%%%%%%%%%%%%%%%%%%%%%%%%%%%%%%%

%%%%%%%%%%%%%%%%%%%%%%%%%%% REML criterion %%%%%%%%%%%%%%%%%%%%%%%%%%%%%%%%%%%%%%%%%%%%%%%%
\subsection{REML criterion}

By integrating the marginal density of $\boldsymbol{\mathcal{Y}}$ with respect to the fixed effects, the REML criterion can be obtained~\citep{laird1982random}. This REML criterion is expressed through the following theorem.

\begin{thm}
Suppose that $\boldsymbol{y}=(y_1^\top,y_2^\top)^\top$ satisfies the bivariate linear mixed-effects model expressed by Equations~(\ref{eq:model_rewriten1} and \ref{eq:model_rewriten2}). Taking into account the notations in the Theorem~\ref{theorem1}, the REML criterion of $\boldsymbol{\sigma}$, $\boldsymbol{\theta}$ and $\rho$ given $\boldsymbol{y}$ is expressed as

\begin{equation}
\mathscr{L}(\boldsymbol{\sigma},\boldsymbol{\theta},\rho|\boldsymbol{y})=\frac{\exp\left[-\frac{r(\widehat{\boldsymbol{\beta}}_{\boldsymbol{\theta},\rho,\boldsymbol{\sigma}},\mu_{\mathcal{U}|\mathcal{Y}=\boldsymbol{y}})}{2\sigma_1^2\sigma_2^2}\right](\sigma_1^2\sigma_2^2)^{\frac{p+q-N}{2}}}{(2\pi)^{(2N-p)/2}|\Sigma_{\boldsymbol{u}}|^{1/2}|L_{\boldsymbol{\theta},\rho,\boldsymbol{\sigma}}||R_X|}.
\end{equation}
\end{thm}

\begin{proof}
\begin{eqnarray}
\mathscr{L}(\boldsymbol{\sigma},\boldsymbol{\theta},\rho|\boldsymbol{y})&=&\int_{\mathbb{R}^p} f_{\boldsymbol{\mathcal{Y}}}(\boldsymbol{y})d\boldsymbol{\beta}\\
&=&\frac{\exp\left[-\frac{r(\widehat{\boldsymbol{\beta}}_{\boldsymbol{\theta},\rho,\boldsymbol{\sigma}},\mu_{\mathcal{U}|\mathcal{Y}=\boldsymbol{y}})}{2\sigma_1^2\sigma_2^2}\right](\sigma_1^2\sigma_2^2)^{\frac{q}{2}}}{(2\pi\sigma_1^2)^{N/2}(2\pi\sigma_2^2)^{N/2}|\Sigma_{\boldsymbol{u}}|^{1/2}|L_{\boldsymbol{\theta},\rho,\boldsymbol{\sigma}}|}\int_{\mathbb{R}^p}\exp\left[-\frac{\left\|R_X(\boldsymbol{\beta}-\widehat{\boldsymbol{\beta}}_{\boldsymbol{\theta},\rho,\boldsymbol{\sigma}})\right\|^2}{2\sigma_1^2\sigma_2^2}\right]d\boldsymbol{\beta}\nonumber\\
&=&\frac{\exp\left[-\frac{r(\widehat{\boldsymbol{\beta}}_{\boldsymbol{\theta},\rho,\boldsymbol{\sigma}},\mu_{\mathcal{U}|\mathcal{Y}=\boldsymbol{y}})}{2\sigma_1^2\sigma_2^2}\right](\sigma_1^2\sigma_2^2)^{\frac{p+q-N}{2}}}{(2\pi)^{(2N-p)/2}|\Sigma_{\boldsymbol{u}}|^{1/2}|L_{\boldsymbol{\theta},\rho,\boldsymbol{\sigma}}||R_X|}\int_{\mathbb{R}^p}(2\pi\sigma_1^2\sigma_2^2)^{-\frac{p}{2}}\exp\left[-\frac{\left\|\boldsymbol{t}\right\|^2}{2\sigma_1^2\sigma_2^2}\right]d\boldsymbol{t},
\end{eqnarray}
where $\boldsymbol{t}=\boldsymbol{\beta}-\widehat{\boldsymbol{\beta}}_{\boldsymbol{\theta},\rho,\boldsymbol{\sigma}}\implies d\boldsymbol{\beta}=\frac{1}{|R_X|}d\boldsymbol{t}$ and $\int_{\mathbb{R}^p}(2\pi\sigma_1^2\sigma_2^2)^{-\frac{p}{2}}\exp\left[-\frac{\left\|\boldsymbol{t}\right\|^2}{2\sigma_1^2\sigma_2^2}\right]d\boldsymbol{t}=1$.
\end{proof}
The REML criterion can also be expressed as

\begin{eqnarray}
\log\left(\mathscr{L}(\boldsymbol{\sigma},\boldsymbol{\theta},\rho|\boldsymbol{y})\right)&=&\tilde{\mathscr{L}}(\boldsymbol{\sigma},\boldsymbol{\theta},\rho|\boldsymbol{y})\nonumber\\
&=&-\frac{r(\widehat{\boldsymbol{\beta}}_{\boldsymbol{\theta},\rho,\boldsymbol{\sigma}},\mu_{\mathcal{U}|\mathcal{Y}=\boldsymbol{y}})}{2\sigma_1^2\sigma_2^2}+\frac{p+q-N}{2}\log(\sigma_1^2\sigma_2^2)-\frac{1}{2}\log(|\Sigma_{\boldsymbol{u}}|)\nonumber\\
&&-\frac{1}{2}\log(|L_{\boldsymbol{\theta},\rho,\boldsymbol{\sigma}}|^2)-\frac{1}{2}\log(|R_X|^2),
\end{eqnarray}
or as

\begin{eqnarray}
-2\tilde{\mathscr{L}}(\boldsymbol{\sigma},\boldsymbol{\theta},\rho|\boldsymbol{y})&=&
\frac{r(\widehat{\boldsymbol{\beta}}_{\boldsymbol{\theta},\rho,\boldsymbol{\sigma}},\mu_{\mathcal{U}|\mathcal{Y}=\boldsymbol{y}})}{\sigma_1^2\sigma_2^2}+(N-p-q)\log(\sigma_1^2\sigma_2^2)+\log(|\Sigma_{\boldsymbol{u}}|)\nonumber\\
&&+\log(|L_{\boldsymbol{\theta},\rho,\boldsymbol{\sigma}}|^2)+\log(|R_X|^2),
\end{eqnarray}
which will be minimized to obtaining the REML estimates of the parameters.
%%%%%%%%%%%%%%%%%%%%%%%%%%% End REML criterion %%%%%%%%%%%%%%%%%%%%%%%%%%%%%%%%%%%%%%%%%%%%%%%%

\section{Simulation studies}

In this Section, the consistency of the estimators is proven through simulation studies, and we compare the present estimation procedure with the EM algorithm. For the sake of simplicity, these simulation studies are performed using simulated bivariate longitudinal data sets. In the following paragraph, we explain how we choose the parameters that have been used to simulate the 'working' data sets.

\paragraph{The working data sets}
We suppose that we are following up a sample of subjects where the goal is to evaluate how the growth of the weight and the height of the individuals of this population are jointly explained by the sex, the score of nutrition (Nscore) and the age. We randomly choose through a uniform distribution the score of nutrition between 20 and 50, and the age between 18 and 37, using the R software. All the computations in this paper are done using the R software. The subject's sex is also randomly chosen. The model under which we simulate the data sets is the following:\\
$n$ indicating the total number of subjects, for $i=1,\dots,n$

\begin{eqnarray}
\text{weight}_i &=&(\mathds{1}_{n_i},\text{sex}_i,\text{Nscore}_i,\text{age}_i)\beta_1+(\mathds{1}_{n_i},\text{Nscore}_i)\gamma_{1i}+\varepsilon_{1i}\nonumber\\
\text{height}_i &=&(\mathds{1}_{n_i},\text{sex}_i,\text{Nscore}_i,\text{age}_i)\beta_2+(\mathds{1}_{n_i},\text{Nscore}_i)\gamma_{2i}+\varepsilon_{2i}
\end{eqnarray}
with
\begin{equation}
\boldsymbol{\gamma}_i=\begin{pmatrix}\gamma_{1i}\\\gamma_{2i}\end{pmatrix}
\sim
\mathcal{N}\left(\boldsymbol{0},\bar{\boldsymbol{\Gamma}}\right),
\varepsilon_{1i}\sim\mathcal{N}\left(0,\sigma_1^2\text{I}_{n_i}\right),
\varepsilon_{2i}\sim\mathcal{N}\left(0,\sigma_2^2\text{I}_{n_i}\right),
\gamma_i\perp\varepsilon_{1i}\perp\varepsilon_{2i}
\end{equation}
The random effect related to the dependent variable 'weight' or 'height' is a vector composed by one random intercept and one random slope in the direction of the covariate 'Nscore'. The total number of observations is denoted by $N$.

We randomly choose $\beta_1,\beta_2,\sigma_1$ and $\sigma_2$ whose values are in the first column of Table~\ref{tab:CmlmeVSem1}. $\bar{\boldsymbol{\Gamma}}$ is also randomly chosen such that it is positive definite, with the following form:

\begin{equation}
\bar{\boldsymbol{\Gamma}}=
\begin{pmatrix}
\eta_1^2&\rho_\eta\eta_1\eta_2&\rho\eta_1\tau_1&\rho\eta_1\tau_2\\
\rho_\eta\eta_1\eta_2&\eta_2^2&\rho\eta_2\tau_1&\rho\eta_2\tau_2\\
\rho\eta_1\tau_1&\rho\eta_2\tau_1&\tau_1^2&\rho_\tau\tau_1\tau_2\\
\rho\eta_1\tau_2&\rho\eta_2\tau_2&\rho_\tau\tau_1\tau_2&\tau_2^2
\end{pmatrix}
\end{equation}
In order to have an almost strong correlation between the marginal random effects, we set $\rho=0.8$ and randomly choose all other parameters involved in the obtaining of $\bar{\boldsymbol{\Gamma}}$. Thus, the obtained $\bar{\boldsymbol{\Gamma}}$ is

\begin{equation}\label{eq:Gamma_vrai}
\bar{\boldsymbol{\Gamma}}=
\begin{pmatrix}
27.77&18.80&41.70&4.93\\
18.80&36.00&47.47&5.62\\
41.70&47.47&97.81&8.91\\
4.93&5.62&8.91&1.37
\end{pmatrix}
\end{equation}

%%%%%%%%%%%%%%%%%%%%%%%%%%% Estimates' performances %%%%%%%%%%%%%%%%%%%%%%%%%%%%%%%%%%%%%%%%%%%%%
\subsection{Estimates' performances}

One practical way to show the consistency of an estimator is by computing its Mean Square Error (MSE). If the MSE of an estimator is asymptotically null, this estimator converges in probability, and is then consistent. In this Section, we gradually simulate data sets with larger sizes , $(N,n)\in\{(600,50),(600,60),\dots,(1000,100),(1000,300),\dots,(15000,1000)\}$. We simulate one hundred data sets of each size and calculate the estimators' MSE using these data sets. This yields one hundred MSE for each size of dataset. This allows to compute the $95\%$ CI (confidence interval) along with the mean of the MSE (one hundred mse) obtained for the hundred data sets of the same size. The results are contained in Table~\ref{tab:mse1} and Table~\ref{tab:mse2}. The Table~\ref{tab:mse1} shows that the asymptotique in the longitudinal data requires not only $n\rightarrow\infty$ and $N\rightarrow\infty$, but also $N/n\rightarrow\infty$. This means that it requires a sufficient total number of observations, a sufficient number of levels for the grouping factor and a sufficient number of observations for each level of the grouping factor. For example, in the Table~\ref{tab:mse1}, when the total number of observations is $N=1000$, the MSE of $\bar{\boldsymbol{\Gamma}}$ is better for $n=100$, $0.47$ $(0.03-1.29)$, than for $n=300$, $0.69$ $(0.06-1.94)$. Observing both Table~\ref{tab:mse1} and Table~\ref{tab:mse2} it is clear that, as the number of observations increase, the MSE descends to 0. We can conclude that the estimators constructed in this paper are consistent. The estimation procedure discussing in this paper may therefore be named Consistent estimates for the Multivariate Linear Mixed-Effects model (Cmlme). The Cmlme acronym will be used in the remainder of the paper for a question of simplicity.

\begin{table}
 \scriptsize
\caption{{\bf Mean Square Error of estimators with $95\%$ CI estimated on $100$ replications for values of $n\in\{50,60,100,300\}$ and $N\in\{600,1000,3000\}$.}}
 \begin{center}
\begin{tabular}[c]{crrrr}
 \toprule
\bf Parameter & $n$ & $N=600$ & $N=1000$ & $N=3000$ \\
 \midrule 
\multirow{4}{1cm}{$\beta_1$}
& 50 & 2.43 (0.11 - 7.11) & 1.89 (0.22 - 4.89) & 1.02 (0.07 - 2.44) \\
& 60 & 2.57 (0.14 - 7.87) & 2.13 (0.26 - 5.55) & 0.77 (0.10 - 2.09) \\
& 100 & 2.27 (0.16 - 5.61) & 1.55 (0.14 - 5.17) & 0.71 (0.04 - 1.85) \\
& 300 & 3.16 (0.14 - 10.54) & 1.70 (0.10 - 4.55) & 0.51 (0.04 - 1.36) \\
\hline
\multirow{4}{1cm}{$\beta_2$}
& 50 & 5.50(0.09 - 16.35) & 3.26 (0.02 - 11.64) & 2.06 (0.09 - 5.98) \\
& 60 & 5.06 (0.12 - 15.09) & 3.22 (0.02 - 10.24) & 1.78 (0.10 - 5.62) \\
& 100 & 4.33 (0.03 - 13.17) & 2.37 (0.02 - 6.89) & 1.06 (0.02 - 3.60) \\
& 300 & 4.58 (0.18 - 15.92) & 2.43 (0.05 - 7.39) & 0.90 (0.04 - 2.88) \\
\hline
\multirow{4}{1cm}{$\sigma_1$}
& 50 & 0.03 (0.00 - 0.14) & 0.02 (0.00 - 0.09) & 0.00 (0.00 - 0.02) \\
& 60 & 0.04 (0.00 - 0.11) & 0.02 (0.00 - 0.08) & 0.00 (0.00 - 0.02) \\
& 100 & 0.03 (0.00 - 0.13) & 0.01 (0.00 - 0.06) & 0.00 (0.00 - 0.01) \\
& 300 & 0.05 (0.00 - 0.23) & 0.03 (0.00 - 0.13) & 0.00 (0.00 - 0.02) \\
\hline
\multirow{4}{1cm}{$\sigma_2$}
& 50 & 0.04 (0.00 - 0.15) & 0.03 (0.00 - 0.08) & 0.00 (0.00 - 0.02) \\
& 60 & 0.04 (0.00 - 0.18) & 0.03 (0.00 - 0.12) & 0.01 (0.00 - 0.03) \\
& 100 & 0.06 (0.00 - 0.18) & 0.03 (0.00 - 0.11) & 0.01 (0.00 - 0.03) \\
& 300 & 0.13 (0.00 - 0.45) & 0.04 (0.00 - 0.15) & 0.01 (0.00 - 0.05) \\
\hline
\multirow{4}{1cm}{$\boldsymbol{\bar{\Gamma}}$}
& 50 & 0.90 (0.12 - 2.41) & 0.62 (0.06 - 1.41) & 0.45 (0.04 - 1.14) \\
& 60 & 1.07 (0.06 - 2.64) & 0.68 (0.08 - 1.98) & 0.25 (0.03 - 0.66) \\
& 100 & 0.90 (0.05 - 2.40) & 0.47 (0.03 - 1.29) & 0.21 (0.02 - 0.71) \\
& 300 & 1.45 (0.19 - 4.39) & 0.69 (0.06 - 1.94) & 0.17 (0.02 - 0.57) \\
\bottomrule
\end{tabular}
 \end{center}
 \label{tab:mse1}
\end{table} 

\begin{table}
 \scriptsize
\caption{{\bf Mean Square Error of estimators with $95\%$ CI estimated on $100$ replications for values of $(n,N)\in\{(500,7000),(600,8000),(800,10000),(1000,15000)\}$.}}
 \begin{center}
\begin{tabular}[c]{ccccc}
 \toprule
\bf Parameter & $n=500,N=7000$ & $n=600,N=8000$ & $n=800,N=10000$ & $n=1000,N=15000$ \\
 \midrule 
\multirow{1}{1cm}{$\beta_1$}
& 0.22 (0.01 - 0.62) & 0.17 (0.01 - 0.48) & 0.16 (0.01 - 0.47) & 0.11 (0.00 - 0.32) \\
\hline
\multirow{1}{1cm}{$\beta_2$}
& 0.32 (0.00 - 1.01) & 0.34 (0.01 - 1.16) & 0.25 (0.02 - 0.67) & 0.21 (0.00 - 0.69) \\
\hline
\multirow{1}{1cm}{$\sigma_1$}
& 0.00 (0.00 - 0.00) & 0.00 (0.00 - 0.01) & 0.00 (0.00 - 0.00) & 0.00 (0.00 - 0.00) \\
\hline
\multirow{1}{1cm}{$\sigma_2$}
& 0.00 (0.00 - 0.01) & 0.00 (0.00 - 0.01) & 0.00 (0.00 - 0.01) & 0.00 (0.00 - 0.00) \\
\hline
\multirow{1}{1cm}{$\boldsymbol{\bar{\Gamma}}$}
& 0.09 (0.00 - 0.25) & 0.07 (0.00 - 0.19) & 0.06 (0.00 - 0.19) & 0.03 (0.00 - 0.09) \\
\bottomrule
\end{tabular}
 \end{center}
 \label{tab:mse2}
\end{table}

%%%%%%%%%%%%%%%%%%%%%%%%%%% End Estimates' performances %%%%%%%%%%%%%%%%%%%%%%%%%%%%%%%%%%%%%%%%%%%

%%%%%%%%%%%%%%%%%%%%%%%%%%% Comparison with EM-based estimates %%%%%%%%%%%%%%%%%%%%%%%%%%%%%%%%%%%%%%%
\subsection{Comparison with EM-based estimates}

In this Section, we compare the estimation procedure based on EM algorithm with the Cmlme. This comparison is performed regarding the accuracy of the estimates, whether or not the starting values of the two algorithms (EM and Cmlme) are naive or advised. We mean by naive starting values, those which are randomly chosen (without specific control), and by advised starting values those obtained by fitting separately each dimension of the bivariate model. The results of these marginal fitting are, indeed, the advised starting values for the bivariate model estimation procedure. The starting values are the same for both Cmlme and EM algorithms. The number of iteration required for convergence, for each algorithm, is also discussed. Our methodology consists in simulating thirty longitudinal data sets of size $(N=3000,n=300)$ and fit the model to each of these data sets using the EM algorithm and the Cmlme, respectively. This allows to compute both the $95\%$ CI and the empirical mean of the thirty estimates in each case (naive and advised starting values). The obtained results are in Table~\ref{tab:CmlmeVSem1} and Table~\ref{tab:CmlmeVSem2}. Table~\ref{tab:CmlmeVSem1} contains the empirical means of the estimates with their $95\%$ CI, and the minimum, the maximum and the average number of iterations. Table~\ref{tab:CmlmeVSem2} contains the empirical relative error of the estimators with their $95\%$ CI. These results show that in the case of naive initialization, the Cmlme estimators outperform the EM estimators. For example, the component of $\beta_1$ which is 14.00 is well estimated by Cmlme, $14.02$ $(13.27-14.45)$ with an empirical relative error of $0.02$ $(0.00-0.04)$, but poorly estimated by EM, $-2.05$ $(-4.70--0.40)$ with an empirical relative error of $1.14$ $(1.01-1.32)$. In the case of advised initialization, both Cmlme and EM algorithms perform well, but Cmlme converge faster (64 iterations in average) than EM (169 iterations in average). The number of iteration required by Cmlme whith advised initializations range from 48 to 89 and from 56 to 103 for naive initializations. The Cmlme with naive initialization therefore needs more iterations than Cmlme with advised initialization for converge. This is expected and may be explained by the fact that the advised initialization values contain some information from the data of interest, and the naive starting points do not.

\begin{table}[!h]
\scriptsize
 \tiny
 \caption{{\bf EM estimates compared with Cmlme estimates on the same data sets. Empirical estimates with their 95\% CI and the number of iteration required for convergence.}}
  \begin{center}
\rotatebox{90}{
\begin{tabular}[c]{crccrccccc}
   \toprule
& & \multicolumn{4}{c}{ \bf Naive initialization} & 
  \multicolumn{4}{c}{ \bf Advised initialization} \\
  \cmidrule(r){3-6}\cmidrule(r){7-10}
&  &\multicolumn{2}{c}{\bf Cmlme}&\multicolumn{2}{c}{\bf EM}
  &\multicolumn{2}{c}{\bf Cmlme}&\multicolumn{2}{c}{\bf EM}\\
  \cmidrule(r){3-4}\cmidrule(r){5-6}\cmidrule(r){7-8}\cmidrule(r){9-10}
 \bf Parameter &\bf Value &\bf Emp. Mean &\bf 95\% CI &\bf Emp. Mean &\bf 95\% CI
 &\bf Emp. Mean &\bf 95\% CI &\bf Emp. Mean &\bf 95\% CI\\
   \midrule
  \multirow{3}{0.1cm}{$\beta_1$}
  &50.67&50.79&$49.14-52.11$&13.47&$-76.70-43.37$&50.80&$49.15-52.12$&50.78&$49.09-52.01$ \\
  &-4.80&$-5.00$&$-8.39--3.66$&-4.79&$-8.08--3.42$&-5.02&$-8.39--3.66$&-4.98&$-8.38--3.65$ \\
  &14.00&14.02&$13.27-14.45$&-2.05&$-4.70--0.40$&14.02&$13.28-14.45$&14.02&$13.27-14.45$ \\
  &2.70&2.70&$2.66-2.72$&2.69&$2.66-2.72$&2.70&$2.66-2.72$&2.70&$2.66-2.72$ \\
  \hline
  \multirow{3}{0.1cm}{$\beta_2$}
  &13.20&13.65&$11.79-15.06$&-84.47&$-114.28--50.63$&13.65&$11.79-15.07$&13.68&$11.81-15.14$ \\
  &-2.80&-2.80&$-4.74--0.43$&-2.75&$-4.90-0.21$&-2.81&$-4.79--0.43$&-2.85&$-4.80--0.51$ \\
  &27.00&27.00&$26.87-27.10$&0.90&$-1.62-2.68$&27.00&$26.87-27.10$&27.00&$26.87-27.10$ \\
  &1.70&1.68&$1.64-1.71$&1.68&$1.64-1.71$&1.68&$1.64-1.71$&1.68&$1.64-1.71$ \\
  \hline
  $\sigma_1$
  &5.80&5.79&$5.62-5.92$&5.78&$5.64-5.98$&5.78&$5.64-5.92$&5.79&$5.65-5.94$ \\
   \hline
  $\sigma_2$
  &7.60&7.61&$7.34-7.74$&7.59&$7.33-7.73$&7.61&$7.34-7.73$&7.63&$7.36-7.73$ \\
  \hline \hline
  \multirow{3}{1cm}{Nbr. of iteration}
  &Min&56&-&63&-&48&-&14&- \\
  &Mean&71&-&109&-&64&-&169&- \\
  &Max&103&-&157&-&89&-&645&- \\
  \bottomrule
  \end{tabular}
  }
   \end{center}
   \label{tab:CmlmeVSem1}
\end{table}

\begin{table}[h!]
 \scriptsize
 \caption{{\bf EM estimates compared with Cmlme estimates on the same data sets. Empirical relative error of the estimates with their 95\% CI}}
  \begin{center}
%\rotatebox{90}{
\begin{tabular}[c]{crccrccccc}
   \toprule
& & \multicolumn{4}{c}{ \bf Naive initialization} & 
  \multicolumn{4}{c}{ \bf Advised initialization} \\
  \cmidrule(r){3-6}\cmidrule(r){7-10}
&  &\multicolumn{2}{c}{\bf Cmlme}&\multicolumn{2}{c}{\bf EM}
  &\multicolumn{2}{c}{\bf Cmlme}&\multicolumn{2}{c}{\bf EM}\\
  \cmidrule(r){3-4}\cmidrule(r){5-6}\cmidrule(r){7-8}\cmidrule(r){9-10}
 \bf Parameter &\bf Value &\bf R. Error &\bf 95\% CI &\bf R. Error &\bf 95\% CI
 &\bf R. Error &\bf 95\% CI &\bf R. Error &\bf 95\% CI\\
   \midrule
  \multirow{3}{0.1cm}{$\beta_1$}
  &50.67&0.01&$0.00-0.03$&0.73&$0.01-2.18$&0.01&$0.00-0.03$&0.01&$0.00-0.03$ \\
  &-4.80&$0.21$&$0.02-0.32$&0.21&$0.00-0.32$&0.21&$0.02-0.32$&0.21&$0.00-0.33$ \\
  &14.00&0.02&$0.00-0.04$&1.14&$1.01-1.32$&0.02&$0.00-0.04$&0.02&$0.00-0.04$ \\
  &2.70&0.00&$0.00-0.01$&0.00&$0.00-0.01$&0.00&$0.00-0.01$&0.00&$0.00-0.01$ \\
  \hline
  \multirow{3}{0.1cm}{$\beta_2$}
  &13.20&0.07&$0.00-0.14$&7.39&$4.19-9.40$&0.07&$0.00-0.14$&0.07&$0.00-0.14$ \\
  &-2.80&0.43&$0.00-0.84$&0.43&$0.02-1.07$&0.43&$0.00-0.84$&0.43&$0.00-0.81$ \\
  &27.00&0.00&$0.00-0.00$&0.96&$0.88-1.02$&0.00&$0.00-0.00$&0.00&$0.00-0.00$ \\
  &1.70&0.01&$0.00-0.02$&0.01&$0.00-0.02$&0.01&$0.00-0.02$&0.01&$0.00-0.03$ \\
  \hline
  $\sigma_1$
  &5.80&0.01&$0.00-0.02$&0.01&$0.00-0.03$&0.01&$0.00-0.02$&0.01&$0.00-0.03$ \\
   \hline
  $\sigma_2$
  &7.60&0.01&$0.00-0.02$&0.01&$0.00-0.04$&0.01&$0.00-0.02$&0.01&$0.00-0.02$ \\
   \bottomrule
  \end{tabular}
 % }
   \end{center}
   \label{tab:CmlmeVSem2}
\end{table}

The empirical mean of the random effects covariance matrix, $\bar{\Gamma}$, is well estimated with advised initializations:

\begin{equation}
\bar{\Gamma}_{\text{adv}}^{\text{Cmlme}}=
\begin{pmatrix}
25.74 & 16.34 & 35.75 & 4.50 \\ 
16.34 & 34.83 & 43.97 & 5.37 \\ 
35.75 & 43.97 & 82.44 & 8.43 \\ 
4.50 & 5.37 & 8.43 & 1.32
\end{pmatrix}
,\quad
\text{with}
\quad
\sigma_{\Gamma_\text{adv}}^{\text{Cmlme}}=
\begin{pmatrix}
5.73 & 2.74 & 4.45 & 0.61 \\ 
2.74 & 2.59 & 3.96 & 0.44 \\ 
4.45 & 3.96 & 12.07 & 0.80 \\ 
0.61 & 0.44 & 0.80 & 0.11
\end{pmatrix}
\end{equation}

and

\begin{equation}
\bar{\Gamma}_{\text{adv}}^{\text{EM}}=
\begin{pmatrix}
23.66 & 16.45 & 32.60 & 4.45 \\ 
16.45 & 34.82 & 43.39 & 5.39 \\ 
32.61 & 43.39 & 75.16 & 8.66 \\ 
4.45 & 5.39 & 8.65 & 1.31
\end{pmatrix}
,\quad
\text{with}
\quad
\sigma_{\Gamma_\text{adv}}^{\text{EM}}=
\begin{pmatrix}
7.49 & 2.71 & 5.29 & 0.76 \\ 
2.71 & 2.59 & 3.71 & 0.45 \\ 
5.30 & 3.71 & 13.36 & 0.80 \\ 
0.76 & 0.45 & 0.80 & 0.11 
\end{pmatrix}
\end{equation}

$\sigma_{\Gamma_\text{adv}}^{\text{Cmlme}}$ and $\sigma_{\Gamma_\text{adv}}^{\text{EM}}$ contain the standard deviations of the entries of $\bar{\Gamma}_{\text{adv}}^{\text{Cmlme}}$ and $\bar{\Gamma}_{\text{adv}}^{\text{EM}}$, respectively. It seems that the empirical standard deviation of the higher entries of $\bar{\Gamma}$ are bigger with EM than with Cmlme. For example, the standard deviation of $\bar{\Gamma}_{11}=27.77$ is $7.49$ for EM, but $5.73$ for Cmlme. Same remark about the standard deviations of $\bar{\Gamma}_{31}$, $\bar{\Gamma}_{32}$ and $\bar{\Gamma}_{33}$, comparing EM and Cmlme. This may be explained by the fact that Cmlme estimators are more consistent than EM estimators. In the case of naive initializations, Cmlme provides a well estimated empirical mean of $\bar{\Gamma}$, when the estimated $\bar{\Gamma}$ provided by EM is very bad (we choose not to show it here).
\newpage
\begin{equation}
\bar{\Gamma}_{\text{naiv}}^{\text{Cmlme}}=
\begin{pmatrix}
24.76 & 16.18 & 34.43 & 4.47 \\ 
16.18 & 34.77 & 43.91 & 5.35 \\ 
34.43 & 43.91 & 80.67 & 8.40 \\ 
4.47 & 5.35 & 8.40 & 1.32
\end{pmatrix}
,\quad
\text{with}
\quad
\sigma_{\Gamma_\text{naiv}}^{\text{Cmlme}}=
\begin{pmatrix}
7.58 & 2.67 & 6.76 & 0.61 \\ 
2.67 & 2.54 & 3.91 & 0.42 \\ 
6.76 & 3.91 & 13.37 & 0.80 \\ 
0.61 & 0.42 & 0.80 & 0.11
\end{pmatrix}
\end{equation}

$\bar{\Gamma}_{\text{naiv}}^{\text{Cmlme}}$ compared to $\bar{\Gamma}_{\text{adv}}^{\text{Cmlme}}$ and $\sigma_{\Gamma_\text{naiv}}^{\text{Cmlme}}$ compared to $\sigma_{\Gamma_\text{adv}}^{\text{Cmlme}}$ show a slight difference which reveals a tiny sensibility of Cmlme to the starting values. This may be corrected by doing more than one evaluation of the model's deviance.

For all the simulation studies, we use the ML deviance criterion (Equation~\ref{eq:devianceML}) and have minimized it using the nlminb function under R software. Thus, the estimates obtained are from the ML estimators. In this paper, we do not provide an application of REML estimates.
%%%%%%%%%%%%%%%%%%%%%%%%%%% End Comparison with EM-based estimates %%%%%%%%%%%%%%%%%%%%%%%%%%%%%%%%%%%%%

%%%%%%%%%%%%%%%%%%%%%%%%%%% Application on real life data %%%%%%%%%%%%%%%%%%%%%%%%%%%%%%%%%%%%%

\section{Application on malaria dataset}

\subsection{Data description}
The data that we analyze here come from a study which was conducted in 9 villages (Avam\'e centre, Gb\'edjougo, Houngo, Anavi\'e, Dohinoko, Gb\'etaga, Tori Cada Centre, Z\'eb\`e and Zoungoudo) of Tori Bossito area (Southern Benin), where \textit{P. falciparum} is the commonest species in the study area ($95\%$) \cite{djenontin2010culicidae} from June $2007$ to January $2010$. The aim of this study was to evaluate the determinants of malaria incidence in the first months of life of child in Benin.

Mothers ($n=620$) were enrolled at delivery and their newborns were actively followed-up during the first year of life. One questionnaire was conducted to gather information on women's characteristics (age, parity, use of Intermittent Preventive Treatment during pregnancy (IPTp) and bed net possession) and on the course of their current pregnancy. After delivery, thick and thin placental blood smears were examined to detect placental infection defined by the presence of asexual forms of \textit{P. falciparum}. Maternal peripheral blood as well as cord blood were collected. At birth, newborn's weight and length were measured and gestational age was estimated.

During the follow-up of newborns, axillary temperature was measured weekly. In case of temperature higher than $37.5^\circ$C, mothers were told to bring their children to the health center where a questionnaire was filled out. A rapid diagnostic test (RDT) for malaria was performed and a thick blood smear (TBS) made. Symptomatic malaria cases, defined as fever ($>37.5^\circ$C) with TBS and/or RDT positive, were treated with an artemisinin-based combination. Systematically, TBS were made every month to detect asymptomatic infections. Every three months, venous blood was sampled to quantify the level of antibody against malaria promised candidate vaccine antigens. Finally, the environmental risk of exposure to malaria was modeled for each child, derived from a statistical predictive model based on climatic, entomological parameters, and characteristics of children's immediate surroundings. Also every $3$ months (at $3, 6, 9, 12, 15, 18$ months 130 of age), infant blood samples were collected.

Concerning the antibody quantification, two recombinant \textit{P. falciparum} antigens where used to perform IgG subclass (IgG1 and IgG3) antibody. Recombinants antigens MSP2 (3D7 and FC27) were from La Trobe University \citep{anders2010recombinant, mccarthy2011phase}. GLURP-R0 (amino acids 25-514, F32 strain) and GLURP-R2 (amino acids 706-1178, 140 F32 strain) were also expressed. The antibodies were quantified in plasma at different times and ADAMSEL FLPb039 software (\url{http:// www.malariaresearch.eu/content/software}) was used to analyze automatically the ELISA optical density (OD) leading to antibody concentrations in ($\mu$g/mL).

In this paper, we use some of the data and we rename the proteins used in the study, for reasons of the protection of these data. Thus, the proteins we use here, are named A1, A2, B and C, and are related to the antigens IgG1 and IgG3 as mentioned above. Information contained in the multivariate longitudinal dataset of malaria are described in the Table~\ref{tab:longiVariables}, where Y denotes an antigen which is one of the following:
\begin{equation}\label{eq:antigens}
\text{IgG1}\_\text{A1},\text{IgG3}\_\text{A1}, \text{IgG1}\_\text{A2},\text{IgG3}\_\text{A2}, \text{IgG1}\_\text{B},\text{IgG3}\_\text{B}, \text{IgG1}\_\text{C},\text{IgG3}\_\text{C}
\end{equation}

\begin{table}[h!]
\caption{{\bf  Variables present in the empirical dataset}}
 \centering
\begin{tabular}[c]{lll}
 \toprule
\bf N$^\circ$ &\bf Variable & \bf Description \\
 \midrule 
1 & id & Child ID \\
2 & conc.Y & concentration of Y \\
3 & conc$\_$CO.Y & Measured concentration of Y in the umbilical cord blood \\
4 & conc$\_$M3.Y & Predicted concentration of Y in the child's peripheral blood at 3 months \\
5 & ap & Placental apposition \\
6 & hb & Hemoglobin level \\
7 & inf$\_$trim & Number of malaria infections in the previous 3 months \\
8 & pred$\_$trim & Quarterly average number of mosquitoes child is exposed to \\
9 & nutri$\_$trim & Quarterly average nutrition scores \\
\bottomrule
\end{tabular}
 \label{tab:longiVariables}
\end{table} 

\subsection{Data analysis}

The aim of the analysis of these data is to evaluate the effect of the malaria infection on the child's immune acquisition (against malaria). Since the antigens which characterize the child's immune status interact together in the human body, we analyze the characteristics of the joint distribution of these antigens, conditionally to the malaria infection and other factors of interest. The dependent variables are then provided by conc.Y (Table~\ref{tab:longiVariables}) which describes the level of the antigen Y in the children at 3, 6, 9, 12, 15 and 18 months. All other variables in the Table~\ref{tab:longiVariables} are covariates. We then have 8 dependent variables which describe the longitudinal profile (in the child) of the proteins listed in Equation~\ref{eq:antigens}.

To illustrate the stability of our approach, we are fitting here a bivariate model to the data, with $\text{IgG1}\_\text{A1}$ and $\text{IgG3}\_\text{A2}$ as dependent variables:

\begin{eqnarray}
	\text{conc.IgG}1\_\text{A}1&=&(\mathds{1},\text{ap},\text{conc$\_$CO.IgG1}\_\text{A}1,\text{conc$\_$M3.IgG1}\_\text{A}1,\text{hb},\text{inf$\_$trim},\nonumber\\
	&&\text{pred$\_$trim},\text{nutri$\_$trim})\beta_1
	+(\mathds{1},\text{inf$\_$trim})\gamma_{1} +\varepsilon_{1}\nonumber\\\nonumber \\
	\text{conc.IgG}3\_\text{A}2&=&(\mathds{1},\text{ap},\text{conc$\_$CO.IgG3}\_\text{A}2,\text{conc$\_$M3.IgG3}\_\text{A}2,\text{hb},\text{inf$\_$trim},\nonumber\\
	&&\text{pred$\_$trim},\text{nutri$\_$trim})\beta_2
	+(\mathds{1},\text{inf$\_$trim})\gamma_{2} +\varepsilon_{2},\label{eq:modelA1B}
\end{eqnarray}
with
\begin{equation}
\boldsymbol{\gamma}=(\gamma_1^\top,\gamma_2^\top)^\top
\sim\mathcal{N}\begin{pmatrix}\bold{0} & \boldsymbol{\bar{\Gamma}}\end{pmatrix},
\quad
\boldsymbol{\varepsilon}=(\varepsilon_1^\top,\varepsilon_2^\top)^\top
\sim\mathcal{N}\left(\bold{0},\begin{pmatrix}\sigma_1^2\text{I} & 0 \\ 0 & \sigma_2^2\text{I}\end{pmatrix}\right).\label{eq:modelA1Bbis}
\end{equation}

Our strategy is to 1) fit the model to the data by running the Cmlme algorithm using 25 different naive starting points and 2) retain the estimates related to the best likelihood (the minimum of the 25 deviances) as the true parameters and compute the estimators' MSE using the 24 others estimates. This may allows to evaluate how much the Cmlme algorithm is sensitive to the starting points. The results are contained in the Table~\ref{tab:empiricalDataEstimation}.

\begin{table}[h!]
\caption{{\bf Empirical data analysis.}}
 \centering
\begin{tabular}[c]{lrrrr}
 \toprule
& \multicolumn{4}{c}{\bf Response variables} \\
 \cmidrule(r){2-5}
&\multicolumn{2}{c}{\bf $\text{conc.IgG}1\_\text{A}1$} &
\multicolumn{2}{r}{\bf $\text{conc.IgG}3\_\text{A}2$}\\
\cmidrule(r){2-3}\cmidrule(r){4-5}
 \bf Covariates&\bf Estimate& \bf MSE &\bf Estimate& \bf MSE\\
 \midrule 
Intercept & $0.609$ & $9.05\times 10^{-5}$ & $-1.626$ & $3.18\times 10^{-5}$ \\
ap & $-0.093$ & $1.06\times 10^{-5}$ &$-0.337$& $1.04\times 10^{-6}$ \\
$\text{conc$\_$CO.IgG1}\_\text{A}1$ & $0.160$ & $1.68\times 10^{-6}$ & $-$ & $-$ \\
$\text{conc$\_$M3.IgG1}\_\text{A}1$ & $0.148$ & $9.85\times 10^{-6}$ & $-$ & $-$  \\
$\text{conc$\_$CO.IgG3}\_\text{A}2$ & $-$ & $-$ & $0.047$ & $6.44\times 10^{-7}$ \\
$\text{conc$\_$M3.IgG3}\_\text{A}2$ & $-$ & $-$ & $0.155$ & $2.22\times 10^{-7}$ \\
hb & $ -0.162$ & $3.22\times 10^{-7}$ & $-0.345$ & $1.35\times 10^{-7}$ \\
$\text{inf$\_$trim}$ & $0.369$ & $1.89\times 10^{-6}$ & $0.696$ & $5.09\times 10^{-7}$ \\
$\text{pred$\_$trim}$ & $-0.003$ & $5.25\times 10^{-8}$ & $ 0.017$ & $1.49\times 10^{-8}$ \\
$\text{nutri$\_$trim}$ & $0.024$ & $5.81\times 10^{-6}$ & $0.115$ & $3.71\times 10^{-5}$ \\
\hline
\multirow{1}{1.6cm}{$\sigma_1$ and $\sigma_2$}
& $1.395$ & $4.96\times 10^{-6}$ & 1.626 & $2.42\times 10^{-5}$ \\
\bottomrule
\end{tabular}
 \label{tab:empiricalDataEstimation}
\end{table} 

Based on these results, the influence of the starting points on the Cmlme algorithm is very low (see the MSE in Table~\ref{tab:empiricalDataEstimation}). The estimated random effects covariance matrix is

\begin{equation}
\boldsymbol{\Gamma}=
\begin{pmatrix}
0.58 & -0.13 & 0.74 & -0.36 \\ 
-0.13 & 0.23 & -0.39 & 0.37 \\ 
0.74 & -0.39 & 0.94 & -0.24 \\ 
-0.36 & 0.37 & -0.24 & 0.34
\end{pmatrix}
\end{equation}

with an MSE of $0.0095$.

%%%%%%%%%%%%%%%%%%%%%%%%%%% End Application %%%%%%%%%%%%%%%%%%%%%%%%%%%%%%%%%%%%%

%%%%%%%%%%%%%%%%%%%%%%%%%%% Conclusion %%%%%%%%%%%%%%%%%%%%%%%%%%%%%%%%%%%%%%%%%%%%%%%%%%%%%
\section{Conclusion}
In the context of fitting multivariate linear mixed-effects model having homoscedastic dimensional residuals, we have suggested ML and REML estimation strategies by profiling the model's deviance and Cholesky factorizing  the random effect covariance matrix. This approach can be considered as the generalization of the approach used by~\cite{bates2014lme4} in the R software lme4 package. Through extensive simulation studies, we have illustrated that the present approach outperforms the traditional EM estimates and provides estimates that are consitent for both fixed effects and variance components. Another interesting characteristic is its robustness relative to the initial value of the optimization procedure which can be randomly chosen without affecting the estimation results. Furthermore, the profiled ML or REML criterion's optimization can be easily and rapidly performed using an existing optimizer in the R software. Further considerations of this approach may include heteroscedastic residuals as well as residuals correlated with random effects, where the theoretical consistency of the resulting estimators will be demonstrated.
%%%%%%%%%%%%%%%%%%%%%%%%%%% End Conclusion  %%%%%%%%%%%%%%%%%%%%%%%%%%%%%%%%%%%%%%%%%%%%%%%%%%
\newpage
\section*{References}

\bibliography{\jobname}
%\bibliography{mybibfile}

\end{document}